\documentclass[prd,nofootinbib,preprintnumbers,preprint]{revtex4}
\usepackage{amsmath}
\usepackage{amsfonts}
\usepackage{amssymb}
\usepackage{dsfont}
\usepackage[hyperfootnotes=false]{hyperref}
\usepackage[dvipsnames]{xcolor}         % default, citations in superscript

\newcommand{\mX}{\mathcal X}

\newcommand{\mK}{\mathcal K}
\newcommand{\mG}{\mathcal G}
\newcommand{\RP}{\right\rangle}
\newcommand{\LP}{\left\langle}
\newcommand{\LCM}{\nabla}
\newcommand{\LCS}{\mathcal D}
\newcommand{\SFS}{{\rm II}}
\newcommand{\MCS}{{\rm H}}
\newcommand{\SFIN}{{\rm II}'}

\newcounter{example}[section]
\newcounter{remark}[section]
\newcounter{theorem}[section]
\newcounter{proposition}[section]
\newcounter{lemma}[section]
\newcounter{corollary}[section]
\newcounter{definition}[section]

\def\theremark{\arabic{section}.\arabic{remark}}
\def\thetheorem{\arabic{section}.\arabic{theorem}}

\def\thedefinition{\arabic{section}.\arabic{definition}}

\renewcommand*{\email}[1]{\footnote{Electronic address: \href{mailto:#1}{\nolinkurl{#1}} }}

\newenvironment{proof}{\noindent {\textit{Proof:}}
}{$\Box$\medskip}

\newenvironment{proposition}{\refstepcounter{theorem}\medskip\noindent{\bf
Proposition \thetheorem}:}{$\Box$\medskip}

\newenvironment{definition}{\refstepcounter{definition}\medskip\noindent{\bf
Definition \thedefinition}:}{$\Box$\medskip}

\begin{document}

\title{ 
Photon regions in stationary axisymmetric spacetimes and umbilic conditions}
\author{Kirill Kobialko\email{kobyalkokv@yandex.ru}}
\author{Dmitri Gal'tsov\email{galtsov@phys.msu.ru}}
\affiliation{Faculty of Physics, Moscow State University, 119899, Moscow, Russia}

\begin{abstract}
We present the fundamentals of the recently proposed geometric description \cite{Kobialko:2020vqf} of {\em photon regions} in terms of foliation into {\em fundamental photon  hypersurfaces}, which satisfies the umbilic condition for the subbundle of the tangent bundle defined by the generalized impact parameter. 

\end{abstract}

\maketitle

\setcounter{page}{2}

\setcounter{equation}{0}
\setcounter{subsection}{0}
\setcounter{section}{0}

\section{Introduction}\label{intro}
Formation of shadows and relativistic images of stationary black holes is closely related to {\em photon regions} \cite{Grenzebach:2014fha,Grenzebach:2015oea}, which are defined as  compact domains where photons can travel endlessly without escaping to infinity or disappearing at the event horizon.
Indeed, the boundary of the gravitational shadow corresponds to the set of light rays that inspiral asymptotically onto the part of the spherical surfaces in photon regions on which closed spherical photon orbits are located \cite{Wilkins:1972rs,Teo:2020sey, Dokuchaev:2019jqq}.

Spherical surfaces in the photon region are just as important for determining the shadow of a stationary black hole as the {\em photon surfaces} \cite{Claudel:2000yi,Gibbons} in the static case\footnote{For recent review of strong {\em gravitational lensing} and shadows see \cite{Perlick:2021aok,Cunha:2018acu,Dokuchaev:2019jqq}. }.
Recall that an important property of the photon surfaces is established by the theorem asserting that these are timelike {\em totally umbilic} hypersurfaces $S$ in spacetime. This means that their second fundamental form $\SFS$ \cite{Chen} is proportional to the induced metric:   
\begin{equation}
\SFS(X,Y)=\MCS \LP X,Y\RP, \quad \forall X,Y \in TS.
\label{01}
\end{equation} 
This property can serve as a constructive definition for analyzing photon surfaces instead of solving geodesic equations. It is especially useful in the cases when the geodesic equations are non-separable, and their analytic solution can not be found \cite{Cornish:1996de,Cunha:2016bjh,Semerak:2012dw, Shipley:2016omi,Cunha:2018gql,Cunha:2017eoe}.
 
However, in rotating spacetime such as Kerr, spherical surfaces in the photon  region do not fully satisfy the umbilic condition and have a boundary.  
Such surfaces usually  form a family, parameterized by the value of the azimuthal impact parameter $\rho=L/E$, where $L,E$ are the integrals of motion corresponding to the timelike and azimuthal Killing vector fields \cite{Galtsov:2019bty,Galtsov:2019fzq}. To describe these surfaces and the photon  region geometrically, we introduce the concept of {\em partially umbilic} submanifolds that weaken the condition (\ref{01}). Namely, it is possible to impose the condition (\ref{01}) not on {\em all} vectors from the tangent space $TS$, but only on some subset of $TS$ specified by the azimuthal impact parameter. 
In addition, we must specify the boundary conditions for the submanifolds so that the photon does not escape through them. Together, this leads to the definition of {\em fundamental photon  submanifolds} \cite{Kobialko:2020vqf,Kobialko:2021aqg} - as generalization of {\em fundamental photon orbits} \cite{Cunha:2017eoe}. The slices consisting of the  fundamental photon surfaces form generalized photon  regions. 
 
In this article, we give a concise presentation of the main concepts of the geometric approach to fundamental photon submanifolds and regions. Section \ref{SS1} describes the partition of the tangent space of the manifold into sectors specified by the azimuthal impact parameter of the geodesics $\rho = L/E$.
Then, in the section \ref{SS2}, we introduce the notion of partially umbilic submanifolds on the $\rho$-constrained sector of the tangent space and define the {\em fundamental photon submanifolds}. 
The Section \ref{SS5} contains the geometric definition of the photon region. 

\section{Geodesic classes} 
\label{SS1} 

Let $M$ be a $m$ dimensional Lorentzian manifold \cite{Chen} with scalar product $\LP \;\cdot\; ,\;\cdot\;\RP$, Levi-Civita connection $\LCM$, a tangent bundle $TM$ and supposed to possess two commuting Killing vector fields $\mK_\alpha$ ($\alpha=t,\varphi$) defining a stationary axisymmetric spacetime such that $\mathcal G=\det(\mathcal G_{\alpha\beta})<0$, where $\mathcal G_{\alpha\beta}=\LP\mK_\alpha,\mK_\beta\RP$. 

Let us define a Killing vector field $\hat{\rho} \in \{\mathcal{K}_\alpha\}$ \cite{Kobialko:2021aqg} with index numbering Killing vectors fields of the frame $\{\mathcal{K}_\alpha\}$
\begin{equation}\label{FPS1}
\hat{\rho}=\rho^\alpha\mK_\alpha, \quad  \rho^\alpha=(\rho,1),
\end{equation}     
which is determined by arbitrary constant parameter $\rho$.  In addition, we will introduce a vector field $\hat{\tau}$ in $\{\mathcal{K}_\alpha\}$ orthogonal to $\hat{\rho}$:
\begin{equation}
    \tau^\alpha =\mathcal{G}^{\alpha\lambda}\epsilon_{\lambda\beta}\rho^\beta, \qquad
    \LP\hat{\tau},\hat{\tau}\RP = - \LP\hat{\rho},\hat{\rho}\RP,\qquad
    \LP\hat{\tau},\hat{\rho}\RP = 0,
\end{equation}
where $\epsilon_{\lambda\beta}$ is the two-dimensional Levi-Civita tensor. 

Then there is frame $\{\hat{\tau},e_a\}\cong\hat{\rho}^{\perp}$, such that $\{e_a\}$ ($a=1,m-2$) is a frame in euclidean orthogonal complement $\{\mathcal{K}_\alpha\}^{\perp}$\footnote{Orthogonal complement ${}^{\perp}$ defined in \cite{Chen}. If $\LP \hat{\rho},\hat{\rho}  \RP=0$, then $\hat{\tau}$ and $\hat{\rho}$ are simply proportional and the orthogonal complement $\hat{\rho}^{\perp}$ will contain only one null vector $\hat{\rho}$ and spacelike $e_a$.}.

Let $\gamma$ be some geodesic on $M$, and  $\dot{\gamma}$ denotes the tangent vector field to $\gamma$. Then there is the following general relationship between geodesics $\gamma$ and orthogonal complement $\hat{\rho}^{\perp}$.

\begin{proposition}
At every point $p\in M$ there is a one-to-one correspondence\footnote{Accurate to the geodesics reparametrization.} between geodesics $\gamma$ (with nonzero energy $E\equiv-\LP\mathcal{K}_t,\dot{\gamma}\RP\neq0$) with impact parameter $\rho=-\LP\mathcal{K}_\varphi,\dot{\gamma}\RP/\LP\mathcal{K}_t,\dot{\gamma}\RP$ and tangent vector fields $\mX\in \hat{\rho}^{\perp}$ with $\LP\mathcal{K}_t,\mX\RP\neq0$. 
\label{P1}
\end{proposition}  

\begin{proof}

If geodesic $\gamma$ has $\rho$ as an impact  parameter, then $\rho=-\LP\mathcal{K}_\varphi,\dot{\gamma}\RP/\LP\mathcal{K}_t,\dot{\gamma}\RP$ and multiplying by $\LP\mathcal{K}_t,\dot{\gamma}\RP\neq0$ we get $0=\LP\rho\mathcal{K}_t+\mathcal{K}_\varphi,\dot{\gamma}\RP=\LP\hat{\rho},\dot{\gamma}\RP$ and therefore $\dot{\gamma}\in\hat{\rho}^{\perp}$. If $\LP\hat{\rho},\mX\RP=0$ and $\LP\mathcal{K}_t,\mX\RP\neq0$, then at any point $p\in M$ the any vector $\mX|_p$ is a tangent vector to some geodesic $\gamma$, which always exists and unique at least in some vicinity of $p\in M$ as solution of ODE with initial conditions $\gamma(0)=p$ and $\dot{\gamma}(0)=\mX|_p$. The geodesic $\gamma$ has $\rho$ as an impact  parameter insofar as $0=\LP\hat{\rho},\dot{\gamma}(0)\RP=\LP\rho\mathcal{K}_t+\mathcal{K}_\varphi,\dot{\gamma}(0)\RP$ and we get $\rho=-\LP\mathcal{K}_\varphi,\dot{\gamma}(0)\RP/\LP\mathcal{K}_t,\dot{\gamma}(0)\RP$.
\end{proof}

It is clear that in the general case the Killing vector field $\hat{\rho}$ can be timelike on the some part of the manifold $M$. In this case, its orthogonal complement $\hat{\rho}^\perp$ will not have the Lorentzian signature everywhere, and therefore not all manifolds will be available for null geodesics with a given impact parameter $\rho$. So, our goal is to find the suitable region $\mathcal C\subset M$ in the original manifold $M$ such that $\LP \hat{\rho},\hat{\rho}\RP|_{\mathcal C}\geq0$.

\begin{proposition}
If $\LCM_{\hat{\rho}}\hat{\rho}\neq0$ for all null $\hat{\rho}$, the smooth function $\LP \hat{\rho},\hat{\rho}\RP$ defines $m$ dimensional manifold with boundary 
\begin{equation}
\mathcal C\subset M: \LP \hat{\rho},\hat{\rho}\RP|_{\mathcal C}\geq0,
\end{equation} 
with interior $\mathcal O:\LP \hat{\rho},\hat{\rho}\RP|_{\mathcal O}>0$ and $m-1$ dimensional boundary $\partial \mathcal C:\LP \hat{\rho},\hat{\rho}\RP|_{\partial\mathcal C}=0$ with outward normal
\begin{equation}
\mathcal N=\LCM_{\hat{\rho}}\hat{\rho}, \quad \LP\mathcal N,\hat{\rho}\RP=0.
\label{a2}
\end{equation} 
\end{proposition} 

\begin{proof}
Let $\mX\in TM$ be an arbitrary vector field in $M$, then using Killing equation we get
\begin{align}
\LP\LCM \LP \hat{\rho},\hat{\rho}\RP,\mX\RP\equiv\LCM_{\mX}\LP \hat{\rho},\hat{\rho}\RP&=2\LP \LCM_X\hat{\rho},\hat{\rho}\RP=-2\LP \LCM_{\hat{\rho}}\hat{\rho},\mX\RP,\\ &  \Updownarrow \nonumber\\
 \LCM \LP \hat{\rho},\hat{\rho}\RP&=-2\LCM_{\hat{\rho}}\hat{\rho}.
\label{a3}
\end{align} 
Thus $\LCM \LP \hat{\rho},\hat{\rho}\RP|_{\partial \mathcal C}\neq0$ and the boundary $\partial \mathcal C$ is a hypersurface in $M$ with the normal field $\mathcal N$ proportional to $\LCM_{\hat{\rho}}\hat{\rho}$. Choosing an outward normal, i.e. directed so that the function $\LP\hat{\rho}, \hat{\rho} \RP$ decreases along the $\mathcal N$ we get first condition in (\ref{a2}). Applying Killing equation again we get second condition in (\ref{a2}) since
\begin{align}
\LP\mathcal N,\hat{\rho}\RP=\LP \LCM_{\hat{\rho}}\hat{\rho},\hat{\rho}\RP=-\LP \LCM_{\hat{\rho}}\hat{\rho},\hat{\rho}\RP=0.
\label{a4}
\end{align} 
\end{proof}
 
\begin{definition}
A connected manifold $\mathcal C$ will be called causal region. The manifold $\mathcal O$ will be called accessible region \cite{Kobialko:2020vqf}.
\end{definition}
 
From the point of view of geodesics, and, in particular, the fundamental photon orbits, the region $\mathcal C$ represents an accessible region for the null geodesics motion in some effective potential  \cite{Cunha:2017eoe,LukesGerakopoulos:2012pq}. Physical meaning of the causal region $\mathcal C$ is that any point can be theoretically observable for any observer in the same region. This causal region may contain spatial infinity (if any) and then will be observable for an asymptotic observer. In some cases, several causal regions may exist, while null geodesics with a given $\rho$ cannot connect one to another. The boundary $\partial \mathcal C$ of the causal region is defined as the branch of the solution of the equation $ \LP \hat{\rho}, \hat{\rho}\RP = 0 $ and is the set of turning points of null geodesics.

\section{Photon submanifold}
\label{SS2}
Let $M$ and $S$ be Lorentzian manifolds, of dimension $m$ and $n$ respectively, and $S\rightarrow M$ an isometric embedding  \cite{Chen} defining $S$ as a submanifold\footnote{A hypersurface if $n=m-1$} in $M$. We adopt here the following convention for the second fundamental form  $\SFS$ of the submanifold\cite{Chen}:      
\begin{align}
\LCM_{Y}X=\LCS_YX+\SFS(X,Y), \quad X,Y\in TS,
\label{b56} 
\end{align} 
where $\LCS_X Y\in TS$, $\SFS(X,Y)\in TS^{\bot}$ and $\LCM$ and $\LCS$ are the Levi-Civita connections on $M$ and $S$ respectively.

\begin{definition}
We will call a submanifold $S$ invariant, if the Killing vector fields  $\mK_\alpha$ and $[\mK_\alpha,\mK_\beta]$ in $M$ are tangent vector fields to $S$.
\label{D3}
\end{definition}

For invariant submanifolds the Killing vectors of $M$ will be also the Killing vectors on the submanifold $S$ \cite{Kobialko:2021aqg} and well defined restrictions $\hat{\rho}^{\perp}$ and $\mathcal C$ on $TS$ and $S$ respectively. We now define a weakened version of the standard umbilic condition (\ref{01}) requiring it to be satisfied only for some subbundle $V\subset TS$ in the tangent bundle $TS$.

\begin{definition}
A submanifold $S$ will be called totally $V$ umbilic if \cite{Kobialko:2020vqf} 
\begin{align} \label{eq_chapter4_no_1}
\SFS(X,Y) =\MCS|_V\LP X,Y\RP, \quad \forall X ,Y\in V.
\end{align} 
\end{definition}

In particular, every totally umbilic  submanifold is trivially totally $V$ umbilic for any $V$. We also note that in the general case $\MCS|_V$ appearing in this formula is only part of the mean curvature \cite{Chen} i.e the trace of $\SFS$ on the subbundle $V$. For invariant totally $V$ umbilic submanifolds, an important theorem on the behavior of null geodesics holds, generalizing the classical result \cite{Chen,Claudel:2000yi}.

\begin{proposition} \label{T1}
Any null geodesic $\gamma$ with impact parameter $\rho$ in an invariant Lorentzian submanifold $S\subset \mathcal O$ is a null geodesic in $M$ if and only if $S$ is totally $\hat{\rho}^{\perp}$ umbilic submanifold.
\end{proposition}

\begin{proof} 
Let $S$ be a totally $\hat{\rho}^{\perp}$ umbilic invariant Lorentzian submanifold and $\gamma$ be an arbitrary affinely parameterized null geodesic with impact parameter $\rho$ in $ S $ i.e. $\LCS_{\dot{\gamma}} \dot{\gamma} = 0$ and $\dot{\gamma} \in \hat{\rho}^{\perp}\subset TS$. Then by the Gauss decomposition (\ref{b56}) 
\begin{align}
\LCM_{\dot{\gamma}}\dot{\gamma}=\LCS_{\dot{\gamma}}\dot{\gamma}+\SFS(\dot{\gamma},\dot{\gamma})=\MCS|_{\hat{\rho}^{\perp}}\LP \dot{\gamma},\dot{\gamma}\RP=0,
\end{align} 
consequently $\gamma$ is a null geodesic in $M$. 

Conversely, let every null geodesic $\gamma$ with impact parameter $\rho$ in $S$ be a null geodesic in $M$, then from the Gauss decomposition (\ref{b56})
\begin{align}
\SFS(\dot{\gamma},\dot{\gamma})=0,
\label{a22}
\end{align} 
for any null $ \dot{\gamma} \in \hat{\rho}^{\perp}\subset TS$. Since we limited ourselves to the accessible region $\mathcal O$ we can choose an orthonormal Lorentzian frame $\{\hat{\tau}/||\hat{\tau}||,e_a\}\cong\hat{\rho}^{\perp}\subset TS$ ($a=1,n-2$). Then the equality (\ref{a22}) for null vectors $\dot{\gamma} = \hat{\tau}/||\hat{\tau}||\pm e{}_a$ in the new frame takes the form
 \begin{align}
\SFS(\hat{\tau}/||\hat{\tau}||,\hat{\tau}/||\hat{\tau}||)+\SFS(e_a,e_a)=0, \quad \SFS(\hat{\tau}/||\hat{\tau}||,e_a)=0.
\end{align} 
And for null vectors $\dot{\gamma}=\hat{\tau}/||\hat{\tau}||\pm (e{}_a\pm e{}_b)/\sqrt{2}$
\begin{align}
\SFS(e_a,e_b)=0.
\end{align} 
\end{proof}

Physical meaning of the Proposition \ref{T1} is that the null geodesics with a given impact parameter $\rho$ initially touching the spatial section of the invariant totally $\hat{\rho}^{\perp}$ umbilic submanifold remain on it for an arbitrarily long time, unless of course they leave it across the boundary. This is a well-known property of a photon sphere and its generalization - a photon surface \cite{Claudel:2000yi}. Thus, we obtain a generalization of the classical definition of the photon surfaces to the case of a class of geodesics with a fixed impact parameter.

It is useful to obtain an equation for the second fundamental form of the totally $\hat{\rho}^{\perp}$ umbilic submanifold in the original basis $\left\{\mK_\alpha, e_a\right\}$.  

\begin{proposition} 
For the invariant $\hat{\rho}^{\perp}$ umbilic submanifold, the second fundamental form $\SFS$ in the basis $\{\mK_\alpha,e_a\}$ has the form
\begin{align} \label{eq_chapter4_no_10}
\SFS=\begin{pmatrix}
-\frac{1}{2} \LCM^{\perp}\mG_{\alpha\beta} & -\sum^{m-n}_{A=1} \epsilon_A\LP\LCM_{\xi_A} \mathcal K_\alpha,e_b\RP\xi_A\\
-\sum^{m-n}_{A=1} \epsilon_A\LP\LCM_{\xi_A} \mathcal K_\beta,e_a\RP\xi_A & \MCS|_{\hat{\rho}^\perp} \LP e_a,e_b\RP
\end{pmatrix},
\end{align}
where $\MCS|_{\hat{\rho}^\perp} $ and $\mG_{\alpha \beta} $ satisfy the master equation \cite{Kobialko:2021aqg}
\begin{align} \label{eq_chapter4_no_11}
\rho^{\alpha}\mathcal M_{\alpha\beta}\rho^{\beta}=0, \quad \mathcal M_{\alpha\beta}=\frac{1}{2}\LCM^{\perp}\left(\mathcal G^{-1}\mathcal G_{\alpha\beta}\right)-\MCS|_{\hat{\rho}^\perp} \left(\mathcal G^{-1}\mathcal G_{\alpha\beta}\right),\\  \quad
\mathcal{G}^{\alpha\lambda}\epsilon_{\lambda\beta}\rho^\beta\SFS(\mathcal K_\alpha,e_a)=0, \label{eq_chapter4_no_11_a}
\end{align}
and the derivative along the unit normals $\xi_A\in TS^{\perp}$ ($A=1,m-n$) of the submanifold $S$ is defined as
\begin{align}
\LCM^{\perp}(\;\cdot\;)=\sum^{m-n}_{A=1} \epsilon_A\LCM_{\xi_A}(\;\cdot\;)\xi_A, \quad \epsilon_A\equiv \LP\xi_A,\xi_A\RP.
\end{align}
\end{proposition}

\begin{proof}
From the Killing equation for $X\in \{e_a\}$ we find
\begin{align} \label{eq_chapter2_no_8}
\SFS(X,\mathcal K_\alpha)=\pi^{\perp}(\LCM_{X} \mathcal K_\alpha)=\sum^{m-n}_{A=1} \epsilon_A\LP\LCM_{X} \mathcal K_\alpha,\xi_A\RP\xi_A
=-\sum^{m-n}_{A=1} \epsilon_A\LP\LCM_{\xi_A} \mathcal K_\alpha,X\RP\xi_A.
\end{align} 
Using the Killing equation again we get
\begin{align} 
\LP\LCM_{\mathcal K_\alpha}\mathcal K_\beta,\xi_A\RP=-\LP\LCM_{\xi_A}\mathcal K_\beta,\mathcal K_\alpha\RP=-\LCM_{\xi_A}\LP\mathcal K_\alpha,\mathcal K_\beta\RP-\LP\xi_A,\LCM_{\mathcal K_\beta}\mathcal K_a\RP.
\end{align} 
Then, using the involutivity condition $[\mathcal K_\alpha, \mathcal K_\beta] \in TS$ we finally find
\begin{align} \label{eq_chapter2_no_10}
\SFS(\mathcal K_\alpha,\mathcal K_\beta)=\sum^{m-n}_{A=1}\epsilon_A\LP\LCM_{\mathcal K_\alpha}\mathcal K_\beta,\xi_A\RP\xi_A=-\frac{1}{2}\sum^{m-n}_{A=1}\epsilon_A\LCM_{\xi_A}\LP\mathcal K_\alpha,\mathcal K_\beta\RP\xi_A.
\end{align} 
and
\begin{align} \label{eq_chapter4_no_12}
    \SFS(\hat{\tau},\hat{\tau}) =
    -\frac{1}{2}\tau^\alpha \tau^\beta \LCM^{\perp}\mG_{\alpha\beta} =
    - \rho^\alpha \rho^\beta \left(
      \frac{1}{2} \LCM^{\perp} \mG_{\alpha\beta}
    + \mG^{\lambda\gamma} \epsilon_{\lambda\alpha}\LCM^{\perp}\epsilon_{\gamma\beta}
    \right)=\\
    \frac{1}{2} \rho^\alpha \rho^\beta \left(
    - \LCM^{\perp} \mG_{\alpha\beta}
    + \mG_{\alpha\beta} \LCM^{\perp} \ln \mathcal{G}
    \right).
\end{align}
Substituting this expression into the equation (\ref{eq_chapter4_no_1}), we get (\ref{eq_chapter4_no_11}).
\end{proof}

Alternatively, the master equation can be rewritten as expression for the mean curvature
\begin{align} 
\label{eq_chapter4_no_41}
\MCS|_{\hat{\rho}^\perp}=\frac{1}{2}\LCM^{\perp}\ln\left(\mathcal{G}^{-1}\LP\hat{\rho},\hat{\rho}\RP\right).
\end{align}
If the submanifold under consideration is totally umbilic we'll get $\mathcal{M}_{\alpha\beta}=0$.

The notion of an invariant $\hat{\rho}^{\perp}$ umbilic submanifold is however too general (as is the notion of an umbilic surface by itself \cite{Cao:2019vlu}) and is not yet defined at the boundary of the causal region $\partial \mathcal C$. Generally speaking, these submanifolds are geodesically not complete (in the sense that null geodesics can leave them across the boundary) or have a non-compact spatial section (geodesics can go into the asymptotic region). Moreover, for each $\rho$ there can be an infinite number of them, just as there are an infinite number of umbilic surfaces, but only one photon sphere in the static Schwarzschild \cite{Cederbaum:2019rbv} solution. Therefore,  it is necessary to introduce a more specific definition of fundamental photon submanifolds \cite{Kobialko:2020vqf,Kobialko:2021aqg}.
 
\begin{definition}
A fundamental photon submanifold $S\subset \mathcal C$ is an invariant Lorentzian submanifold with compact spatial section such that:

(a) The boundary $\partial S$ (if any) lie in $\partial \mathcal C$.

(b) The second fundamental form $\SFS$ has the form (\ref{eq_chapter4_no_10}) and satisfies the master equation/inequality
\begin{align} \label{eq_chapter2_no_9}
\mathcal M(\hat{\rho},\hat{\rho})=0, \quad \LP\hat{\rho},\hat{\rho}\RP\geq0.
\end{align} 
\label{D5}
\end{definition}

In the case $n=m-1$, the fundamental photon submanifold is a timelike fundamental photon hypersurface (FPH).
 
\begin{proposition}
Every null geodesic $\gamma$ with an impact parameter $\rho$ at least once touching an arbitrary fundamental photon submanifold $S$ lies in it completely i.e. $\gamma\subset S$.
\label{P3}	
\end{proposition}

\begin{proof}
Obviously condition (b) in Definition \ref{D5}, by virtue of Proposition \ref{T1}, prevents null geodesics from leaving the fundamental photon submanifold at all interior points $(S/\partial S)\cap \mathcal O$. 

Conditions (a-b) for boundary points $\partial S$ and interior points $S \cap \partial \mathcal C$ prevents the possibility of null geodesics to leave fundamental photon submanifolds through them. Indeed, $\partial \mathcal C$ is the set of turning points for null geodesics in $M$ that cannot move in the direction of the normal $\mathcal N|_{\partial S}$, while condition $\SFS(\hat{\rho}, \hat{\rho})|_{\partial S} = 0$ following from (\ref{eq_chapter2_no_9}) not only prevents geodesics from moving in the normal directions $\xi_A$ but also ensures that the normal $\mathcal N_{\partial S}$ to $\partial S$ in $S$ coincides with the normal $\mathcal N|_{\partial S}$ to the $\partial \mathcal C$ in $M$.
\begin{equation} \label{eq_chapter4_no_7}
\mathcal N|_{\partial S}=\LCM_{\hat{\rho}}\hat{\rho}|_{\partial S}=\left\{\LCS_{\hat{\rho}}\hat{\rho}+\SFS(\hat{\rho},\hat{\rho})\right\}|_{\partial S}=\LCS_{\hat{\rho}}\hat{\rho}|_{\partial S}=\mathcal N_{\partial S}.
\end{equation}   
\end{proof}

From this statement, it is clear that the so-defined fundamental photon submanifolds have trapping properties even at the boundary and contain

(a) non-periodic trapped photon  orbits,

(b) periodic fundamental photon orbits \cite{Cunha:2017eoe}.

In stationary and axisymmetric geometry, there exists a vector field 
\begin{equation} \label{eq_chapter4_no_32}
\hat\omega = \omega^\alpha \mathcal K_\alpha, \quad \omega^\alpha=(1,\omega),
\end{equation} 
orthogonal to the all spatial slices $\Sigma$ of manifold $M$\footnote{The norm $||\hat\omega||$ of vector field  $\hat\omega$ is called the lapse function.}. In such slices, one can define fundamental photon submanifolds in terms of {\em principal curvatures}.

\begin{proposition} \label{prop_chapter3_no_7}
The principal curvatures $\lambda$ of the spatial section $S'$ of the invariant totally $\hat{\rho}^{\perp}$ umbilic submanifold $S$ satisfy the master equation
\begin{align} \label{eq_chapter4_no_26}
\lambda_{a}-\lambda_{\varphi}=\LCM^{\perp}\ln(||\hat{\rho}||/||\hat\omega||),
\end{align} 
with 
\begin{align} \label{eq_chapter4_no_25}
\lambda_{\varphi}=-\frac{1}{2} \LCM^{\perp}\ln\mathcal G_{\varphi\varphi}, \quad \lambda_{a}=\MCS|_{\hat{\rho}^{\perp}}.
\end{align} 
\end{proposition} 

\begin{proof}
From the general theory for the second fundamental form $\SFIN$ of the spatial section $S'\subset \Sigma$ of the invariant totally $\hat{\rho}^{\perp}$ umbilic submanifold $S$ we find that
\begin{align}  \label{eq_chapter4_no_22}
&\SFIN(X,Y)=\SFS(X,Y)=\MCS|_{\hat{\rho}^{\perp}}\LP X,Y\RP, \quad \SFIN(\mathcal K_\varphi,\mathcal K_\varphi)=\SFS(\mathcal K_\varphi,\mathcal K_\varphi) =-\frac{1}{2} \LCM^{\perp}\mathcal G_{\varphi\varphi},\nonumber\\
&\SFIN(\mathcal K_\varphi,X)=\SFS(\mathcal K_\varphi,X)=-\sum^{m-n}_{A=1} \epsilon_A\LP\LCM_{\xi_A} \mathcal K_\varphi,X\RP\xi_A, \quad X\in \{e_a\}.\nonumber
\end{align}
In what follows, we will assume that $\SFS(\mathcal K_\varphi,X) = 0$\cite{Kobialko:2020vqf}. Due to the shape of the second fundamental form of the spatial slices $\Sigma$ and the expression for mixed components, we find $\SFS(\hat\omega,X) = 0$, and therefore from $\SFS(\mathcal K_\varphi,X) = 0 $ follows $\SFS(\mathcal K_t,X) = 0$ and the second umbilic condition (\ref{eq_chapter4_no_11_a}) is fulfilled identically. Further, from the master equation (\ref{eq_chapter4_no_41}) and  expressions for the determinant $\mathcal G$
\begin{align} 
\mathcal G=\LP\mathcal K_t,\mathcal K_t\RP \LP\mathcal K_\varphi,\mathcal K_\varphi\RP -\LP\mathcal K_t,\mathcal K_\varphi\RP^2\nonumber\\=(\LP\hat\omega,\hat\omega\RP-2 \omega \LP\mathcal K_t,\mathcal K_\varphi\RP-\omega^2\LP\mathcal K_\varphi,\mathcal K_\varphi\RP)\LP\mathcal K_\varphi,\mathcal K_\varphi\RP-\omega^2 \LP\mathcal K_\varphi,\mathcal K_\varphi\RP^2\nonumber\\=\LP\hat\omega,\hat\omega\RP\LP\mathcal K_\varphi,\mathcal K_\varphi\RP.
\end{align} 
we get (\ref{eq_chapter4_no_26}).
\end{proof}

\section{Photon Region}
\label{SS5}
We now define the concept of a fundamental photon region and a fundamental photon function -- a generalization of the classical three-dimensional photon region in the Kerr metric
\cite{Grenzebach:2014fha,Grenzebach:2015oea,Grover:2017mhm}.

\begin{definition}
The fundamental photon function $PF$ will be called the mapping \cite{Kobialko:2020vqf}
\begin{align}
PF:\rho \rightarrow \bigcup S
\end{align} 
which associates with each $\rho$ one or the union of several fundamental photon  submanifolds with the same $\rho$.
\end{definition}

The function $PF$ can be continuous, defining some connected smooth submanifold in the extended manifold $\left\{M,\rho\right\}$. At the same time, several continuous functions $PF$ can exist in which different FPHs correspond to one $\rho$. In particular, for a given $\rho$, photon  and antiphoton FPHs ((un)stable photon surface \cite{Koga:2019uqd}) can occur simultaneously, indicating the instability of the solution \cite{Gibbons}.  

\begin{definition}
The fundamental photon region is the complete image of the function $PF$
\begin{align}
PR=\bigcup_{\rho} PF.
\end{align} 
\end{definition}

A fundamental photon region is a standard region in the space $M$ in which there are fundamental photon orbits and, in particular, the classical photon region in the Kerr metric. However, the mapping $PF$ can several times cover the image of $PR$ or part of it when the parameter $\rho$ is continuously changed. For example, in the case of a static space, $PR$ is covered at least two times, i.e. $PF$ is a two-sheeted function.

The photon region can also be described using an algebraic equation \cite{Kobialko:2021aqg}. To do this, rewrite the equations and the inequality (\ref{eq_chapter2_no_9}) as
\begin{align} \label{eq_chapter4_no_16}
\mathcal M_{tt}\rho^2+2\mathcal M_{t\varphi}\rho+\mathcal M_{\varphi\varphi}=0,\\
\mathcal G_{tt}\rho^2+2\mathcal G_{t\varphi}\rho+\mathcal G_{\varphi\varphi}\geq0.
\end{align}
Then, excluding $\rho$, we find
\begin{align} \label{eq_chapter4_no_18}
2(-\mathcal M_{t\varphi}\pm\sqrt{\mathcal M_{t\varphi}^2-\mathcal M_{tt}\mathcal M_{\varphi\varphi}})(\mathcal G_{t\varphi}-\mathcal G_{tt}\mathcal M_{t\varphi}/\mathcal M_{tt})/\mathcal M_{tt}\nonumber\\+(\mathcal G_{\varphi\varphi}-\mathcal G_{tt}\mathcal M_{\varphi\varphi}/\mathcal M_{tt})\geq0,
\end{align}
or
\begin{align}\label{PR1b}
\pm 2(\mathcal G_{t\varphi}\mathcal M_{tt}-\mathcal G_{tt}\mathcal M_{t\varphi})\sqrt{-\mathcal M}-2\mathcal G_{tt} \cdot \mathcal M+\mathcal M_{tt} \cdot\mathcal G \cdot {\rm Tr}(\mathcal M) >0.
\end{align}
These inequality describe a generalized photon region.

\setcounter{equation}{0}

\section{Conclusion}
\label{Concl}
In this article, we briefly presented a purely geometric approach to defining characteristic surfaces and regions filled with closed photon  orbits, based on some generalization of umbilic hypersurfaces. The main new concept is a {\em partially} umbilic surface, which has umbilic properties with respect to a correctly defined subbundle of the tangent bundle. This approach does not address the integration of geodesic equations, and thus is applicable to spacestimes with a non-integrable geodesic structure.
 
We tried to give a more clear  and  concise idea of the main geometric notions presented in \cite{Kobialko:2020vqf}, supplementing them with a number of new useful expressions and relations, which, in particular, turned out to be useful for analyzing their connection with Killing tensor fields \cite{Kobialko:2021aqg}. We hope that this formalism will pave the way for obtaining new topological constraints, Penrose-type inequalities (and other estimates) \cite{Shiromizu:2017ego,Feng:2019zzn,Yang:2019zcn},  uniqueness theorems\cite{Cederbaum:2015fra,Yazadjiev:2015hda,Yazadjiev:2015mta,Yazadjiev:2015jza,Rogatko:2016mho,Cederbaum:2019rbv}, similar to ones known for photon spheres and transversally trapping surfaces \cite{Yoshino:2019dty,Yoshino:2019mqw}.

The work is supported by the Russian Foundation for Basic Research on the project 20-52-18012Bulg-a, and the Scientific and Educational School of Moscow State University “Fundamental and Applied Space Research”.

\end{document}